\documentclass{amsart}

\usepackage{amssymb}
\usepackage{amsthm}
\usepackage{amsmath} 
\usepackage{hyperref}
\usepackage{tikz}
\newtheorem{theorem}{Theorem}

\newtheorem{corollary}{Corollary}

\newtheorem{lemma}{Lemma}

\newcommand{\supp}{\operatorname{supp}}

\begin{document}

\title[]{Riesz energy on the torus: \\ regularity of minimizers} \keywords{}

\thanks{The research of J.L. was supported in part by the National Science Foundation under award DMS-1454939.}

\author[]{Jianfeng Lu}
\address[Jianfeng Lu]{Department of Mathematics, Department of Physics, and Department of Chemistry,
Duke University, Box 90320, Durham NC 27708, USA}
\email{jianfeng@math.duke.edu}

\author[]{Stefan Steinerberger}
\address[Stefan Steinerberger]{Department of Mathematics, Yale University, New Haven, CT 06510, USA}
\email{stefan.steinerberger@yale.edu}

\begin{abstract}  We study sets of $N$ points on the $d-$dimensional torus $\mathbb{T}^d$ minimizing interaction functionals of the type
\[
\sum_{i, j =1 \atop i \neq j}^{N}{ f(x_i - x_j)}.
\]
The main result states that for a class of functions $f$ that behave
like Riesz energies $f(x) \sim \|x\|^{-s}$ for $ d-2 < s < d$, the minimizing
configuration of points can be expected to have optimal regularity w.r.t.~a
Fourier-analytic regularity measure that arises in the study of
irregularities of distribution. A particular consequence is that they
are optimal quadrature points in the space of trigonometric
polynomials up to a certain degree. The proof extends to other
settings and also covers less singular functions such as
$f(x) = \exp\bigl(- N^{\frac{2}{d}} \|x\|^2 \bigr)$.
\end{abstract}
\maketitle

\section{Introduction and Main Result}
\subsection{Introduction.} This paper studies the regularity of
minimizers of variational problems. More precisely, for a function
$f:\mathbb{T}^d \rightarrow \mathbb{R}$ we will be interested in
configuration of $N$ points
$\left\{x_1, x_2, \dots, x_N \right\} \subset \mathbb{T}^d$ that
minimize the energy functional
$$ \sum_{i \neq j}{ f(x_i - x_j)}. $$
These questions have a long and rich history: the choice
$f(x_i -x_j) = \|x_i - x_j\|^{-1}$ on $\mathbb{S}^2$ is often
interpreted as the minimal energy configuration of $N$ electrons on a
sphere and dates back to the physicist J. J. Thomson \cite{thomson} in
1904. Minimizers are `roughly' evenly spaced and what remains to be
understood are fine structural details of the minimizing
configuration. These questions are very relevant in mathematical
physics (cf. Abrikosov lattices \cite{abri}). Since the field is
extremely active, it has become increasingly difficult to summarize
existing results, we refer to the survey of Blanc \& Lewin
\cite{blanc} for an introduction into the crystallization conjecture,
to a recent survey of Brauchart \& Grabner \cite{brauchart} for an
introduction to general problems of these type and to recent lecture
notes of Serfaty \cite{serf}.

\subsection{Measuring Regularity.} When studying the regularity of
minimizers, the predominant measures are usually phrased in local
terms: for example, is it true that
$\min_{i \neq j}\|x_i - x_j\| \gtrsim N^{-1/d}$? Since the minimizers
are assumed to be extremely regular and perhaps even close to
lattices, it is reasonable to believe that this is indeed the case;
the first results in this direction are due to Dahlberg
\cite{dahlberg} and this has inspired many subsequent results. The
purpose of our paper is to point out a particular
nonlocal measure of regularity that can be applied to be this problem.
Given a set of $N$ points on the torus
$\left\{x_1, \dots, x_N\right\} \subset \mathbb{T}^d $, we can quantify regularity by the size of the Fourier
coefficients of the sum of Dirac measures placed in these points (since we work on the torus, the Fourier grid is given by $\mathbb{Z}^d$, which we assume without explicit mentioning in the sequel). The quantity we will study is given by
$$  \sum_{\|k \| \leq X \atop k \neq 0}{ \left| \widehat{ \sum_{i=1}^{N}{ \delta_{x_i}}}(k) \right|^2} = \sum_{\|k\| \leq X \atop k \neq 0}{ \left| \sum_{n=1}^{N}{ e^{2 \pi i \left\langle k, x_n \right\rangle}}\right|^2},$$
where $X > 0$ is a free parameter.
It is not new and played a prominent role in the $L^2-$theory of irregularities of distribution, we refer to the seminal work of Beck \cite{beck1,beck2,beck3} and Montgomery \cite{montf, mont1, mont}.
 There is a fundamental inequality of Montgomery \cite{mont1, mont} (see also the refinement \cite{stein1}) that
states that this quantity cannot be too small. More precisely, for all point sets $\left\{x_1, \dots, x_N \right\} \subset \mathbb{T}^d$ and $X \in \mathbb{N}$
\begin{equation} \label{eq:montgomery}
\sum_{\|k\| \leq X \atop k \neq 0}{\left| \sum_{n=1}^{N}{ e^{2 \pi i \left\langle k, x_n \right\rangle}}\right|^2 } \geq N X^d - N^2.
\end{equation}
This inequality is sharp up to constants for sets of points satisfying a separation condition $\|x_i - x_j\| \gtrsim N^{-1/d}$ (see \cite{stein1}). While, conversely, its validity does not imply $\sim N^{-1/d}$ separation,
it does imply global regularity results at scales slightly coarser than that of nearest neighbors (see \S 1.4. for a precise statement).

\subsection{Main Results.} We can now state our main result.
In the formulation of the main result, one should imagine $f$ to grow like
$$ f(x) \sim \|x\|^{-s} \qquad \mbox{for some}~0< s < d.$$
This is not required but coincides with the way the quantities would scale naturally. We first state
the result and then discuss its assumptions in greater detail.
\begin{theorem} Let $f:\mathbb{T}^d \rightarrow \mathbb{R}$ be given and assume
there exists positive $c_1, c_2$ such that for all $x \in \mathbb{T}^d$, $k \in \mathbb{Z}^d$
\begin{equation} \label{eq:condition}
   \frac{c_1}{1 + \|k\|^{d-s}} \leq \widehat{f}(k) \leq \frac{c_2}{1 + \|k\|^{d-s}}
\end{equation}
as well as, for all $t > 0$ and all $x \in \mathbb{T}^d$,
\begin{equation} \label{eq:condition3}
 \int_{\mathbb{T}^d}{ f(x - y) \left[ e^{t\Delta} \delta_{0} \right] (y) dy} \leq f(x) + c_2 t \left| (\Delta f)(y)\right|.
\end{equation}
Suppose furthermore that $\left\{x_1, \dots, x_N\right\} \subset \mathbb{T}^d$ satisfies, 
\begin{equation}
 \sum_{i \neq j}{f(x_i - x_j)} \leq N^2 \int_{\mathbb{T}^d}{\int_{\mathbb{T}^d}{f(x-y) dx}dy} + E.
\end{equation}
Then, for every $c_3 >0$, we have
$$ \sum_{\|k\| \leq c_3 N^{1/d} \atop k \neq 0}{\left| \sum_{n=1}^{N}{ e^{2 \pi i \left\langle k, x_n \right\rangle}}\right|^2 } \lesssim_{c_1, c_2, c_3} N^{1 - \frac{s}{d}}E +  N^{\frac{d-(s+2)}{d}} \sum_{i,j=1 \atop i \neq j}^{N}{\left|  (\Delta f)(x_i - x_j)\right| }.$$
\end{theorem}

Condition \ref{eq:condition} mirrors the classical scaling of the Riesz kernel in $\mathbb{R}^d$. 
Condition \ref{eq:condition3} is essentially controlling convexity and
is fairly easy to satisfy: the scaling is exactly what is required by a Taylor expansion up to second order which can be seen as follows. The heat kernel applied to
a delta mass in the origin up to time $t$ essentially yields a Gaussian at scale $\sim \sqrt{t}$. This means that we would expect
$$  \int_{\mathbb{T}^d}{ f(x - y) \left[ e^{t\Delta} \delta_{0} \right] (y) dy} \sim \frac{1}{\mathbb{S}^{d-1}} \int_{\mathbb{S}^{d-1}}{ f( x + \sqrt{t} z) d\sigma(z)}.$$
A Taylor expansion up to second order suggests
$$  \frac{1}{\mathbb{S}^{d-1}} \int_{\mathbb{S}^{d-1}}{ f( x + \sqrt{t} z) d\sigma(z)} \sim f(x) + t (\Delta f)(x).$$
We will show below, in Lemma 1, that for the classical Riesz potential this estimate actually degenerates in the range $\|x\|^2 \lesssim t$ and that stronger results hold.
The result is most interesting in the range $d-2 < s < d$. It is known (see e.g. \cite{hardin}) that the error term for a near-optimal point configuration scales like
$$ E \sim N^{1 + \frac{s}{d}}.$$

One naturally expects that if
$$ f(x) \sim \frac{1}{\|x\|^{s}} \quad \mbox{then} \quad  \Delta f \sim \frac{1}{\|x\|^{s+2}}$$
which moves the power out of the potential-theoretic regime. This implies
$$ \sum_{i,j=1 \atop i \neq j}^{N}{\left|  (\Delta f)(x_i - x_j)\right| } \sim \sum_{i,j=1 \atop i \neq j}^{N}{ \frac{1}{\|  x_i - x_j \|^{s+2} } } \sim N^{\frac{d + s+2}{d}}$$
which would then imply
$$ \sum_{\|k\| \leq c_4 N^{1/d} \atop k \neq 0}{\left| \sum_{n=1}^{N}{ e^{2 \pi i \left\langle k, x_n \right\rangle}}\right|^2 } \lesssim N^2$$
and would be optimal by Montgomery's Lemma. It seems reasonable to assume that this last statement actually holds for all $0 < s < d$
and all optimizing point configurations but it is not a surprise that the regime $d-2 < s < d$ is slightly better behaved. We do not know
how one would go about establishing the result for other values of $s$. However, the statement that optimal point configurations 
are maximally separated actually implies the exponential sum estimate -- it is therefore strictly simpler and may provide a valuable
stepping stone.
The approach also suggests another natural question.

\begin{quote}
\textbf{Open Problem.}  Let $0< s,t < d$ and $\left\{x_1, \dots, x_N\right\} \subset \mathbb{T}^d$. Does
$$\sum_{i \neq j}{\frac{1}{\|x_i - x_j\|^s}} \leq N^2 \frac{1}{|\mathbb{T}^d|}\int_{\mathbb{T}^d}{\frac{1}{\|x\|^s} dx} + c_1 N^{1+\frac{s}{d}}$$
imply
$$  \sum_{i \neq j}{\frac{1}{\|x_i - x_j\|^t}} \leq N^2 \frac{1}{|\mathbb{T}^d|}\int_{\mathbb{T}^d}{\frac{1}{\|x\|^t} dx} + c_1 N^{1+\frac{t}{d}}?$$
\end{quote}
The question is clearly one possible way of encapsulating the natural notion that minimizing configurations are close to lattices.

\subsection{Less singular kernels.}

 The idea behind the proof is so general that the method is not restricted to these interaction energies;
the following result is another application of this approach.
\begin{theorem}[see \cite{stein1}] Every set of points $\left\{x_1, \dots, x_N\right\} \subset \mathbb{T}^d$ satisfying
 $$ \sum_{i, j =1 }^{N}{ \exp\left( - N^{\frac{2}{d}} \|x_i - x_j\|^2 \right)} \leq c_1 N $$
satisfies
$$ \sum_{\|k\| \leq c_2 N^{1/d} \atop k \neq 0}{\left| \sum_{n=1}^{N}{ e^{2 \pi i \left\langle k, x_n \right\rangle}}\right|^2 } \lesssim_{c_1, c_2} N^2.$$
\end{theorem}

This result, first proved in \cite{stein1}, was what originally motivated our interest in this problem.
If the point set is well-distributed (i.e. $\|x_i - x_j\| \gtrsim N^{-1/d}$ whenever $i \neq j$), then the assumption is easily seen to hold and the conclusion
follows; the interesting part is that any set of points minimizing this Gaussian interaction functional will necessarily behave like a set
of well-separated points w.r.t. this Fourier-analytic regularity measure.
\begin{center}
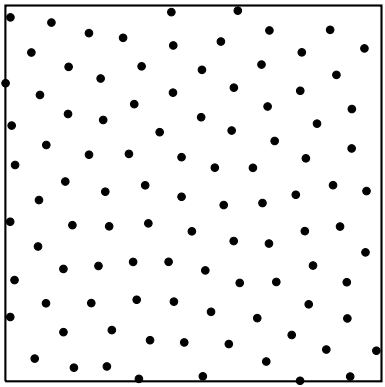
\begin{figure}[h!]
\begin{tikzpicture}[scale=5]
\draw [thick] (0,0) -- (1,0) -- (1,1) -- (0,1) -- (0,0);
\filldraw (0.86344, 0.935215) circle (0.01cm);
\filldraw (0.761266, 0.123258) circle (0.01cm);
 \filldraw (0.312914, 0.914013) circle (0.01cm);
 \filldraw (0.282995, 0.136372) circle (0.01cm);
 \filldraw (0.520495, 0.702722) circle (0.01cm);
 \filldraw (0.889789, 0.411789) circle (0.01cm);
 \filldraw (0.38007, 0.420294) circle (0.01cm);
 \filldraw (0.828529, 0.685654) circle (0.01cm);
 \filldraw (0.0166807, 0.680519) circle (0.01cm);
 \filldraw (0.817869, 0.308049) circle (0.01cm);
 \filldraw (0.957418, 0.343145) circle (0.01cm);
 \filldraw (0.54677, 0.184955) circle (0.01cm);
 \filldraw (0.178159, 0.415532) circle (0.01cm);
 \filldraw (0.693605, 0.0526631) circle (0.01cm);
 \filldraw (0.715936, 0.639526) circle (0.01cm);
 \filldraw (0.441267, 0.982378) circle (0.01cm);
 \filldraw (0.222136, 0.92647) circle (0.01cm);
 \filldraw (0.468352, 0.596433) circle (0.01cm);
 \filldraw (0.580453, 0.469036) circle (0.01cm);
 \filldraw (0.788105, 0.875356) circle (0.01cm);
 \filldraw (0.607445, 0.781374) circle (0.01cm);
 \filldraw (0.025742, 0.575613) circle (0.01cm);
 \filldraw (0.607019, 0.373251) circle (0.01cm);
 \filldraw (0.275914, 0.412401) circle (0.01cm);
\filldraw (0.572954, 0.904129) circle (0.01cm);
 \filldraw (0.669615, 0.168117) circle (0.01cm);
 \filldraw (0.907688,  0.263678) circle (0.01cm);
 \filldraw (0.247459, 0.306925) circle (0.01cm);
 \filldraw (0.683555, 0.474396) circle (0.01cm); 
\filldraw (0.697176, 0.731277) circle (0.01cm);
 \filldraw (0.0779352, 0.0604842) circle (0.01cm);
 \filldraw (0.880141, 0.815375) circle (0.01cm);
 \filldraw (0.339601, 0.317686) circle (0.01cm); 
\filldraw (0.0129042, 0.171138) circle (0.01cm);
 \filldraw (0.68092, 0.842863) circle (0.01cm); 
\filldraw (0.522623, 0.828822) circle (0.01cm);
\filldraw (0.122374, 0.954634) circle (0.01cm);
 \filldraw (0.871111, 0.521851) circle (0.01cm);
 \filldraw (0.259988, 0.695398) circle (0.01cm);
 \filldraw (0.62305, 0.261789) circle (0.01cm);
 \filldraw (0.475526, 0.103445) circle (0.01cm);
 \filldraw (0.920645, 0.619608) circle (0.01cm);
 \filldraw (0.384579, 0.10938) circle (0.01cm); 
\filldraw (0.000487339, 0.793214) circle (0.01cm);
 \filldraw (0.06903, 0.874934) circle (0.01cm);
 \filldraw (0.108716, 0.628824) circle (0.01cm);
 \filldraw (0.354928, 0.00653162) circle (0.01cm); 
\filldraw (0.44633, 0.893699) circle (0.01cm);
 \filldraw (0.269775, 0.039552) circle (0.01cm);
 \filldraw (0.495805, 0.39909) circle (0.01cm);
 \filldraw (0.0867259,0.35875) circle (0.01cm); 
\filldraw (0.916754, 0.0126726) circle (0.01cm);
 \filldraw (0.594041, 0.0992317) circle (0.01cm);
 \filldraw (0.012737, 0.424625) circle (0.01cm);
 \filldraw (0.18215, 0.0362856) circle (0.01cm); 
\filldraw (0.921154, 0.724508) circle (0.01cm);
 \filldraw (0.154361, 0.130826) circle (0.01cm);
 \filldraw (0.986316, 0.0813935) circle (0.01cm); 
\filldraw (0.85324, 0.0845723) circle (0.01cm);
 \filldraw (0.448039,0.212008) circle (0.01cm);
 \filldraw (0.524963, 0.0132488) circle (0.01cm);
 \filldraw (0.0133802, 0.968421) circle (0.01cm);
 \filldraw (0.41018, 0.66293) circle (0.01cm);
 \filldraw (0.168091, 0.836753) circle (0.01cm);
 \filldraw (0.960149, 0.506272) circle (0.01cm); 
\filldraw (0.154239, 0.299023) circle (0.01cm); 
\filldraw (0.772125, 0.49642) circle (0.01cm);
 \filldraw (0.222319, 0.602971) circle (0.01cm);
 \filldraw (0.531493, 0.295142) circle (0.01cm);
 \filldraw (0.265543, 0.504527) circle (0.01cm);
 \filldraw (0.909212, 0.167291) circle (0.01cm);
 \filldraw (0.72026, 0.264402) circle (0.01cm);
 \filldraw (0.228343, 0.208178) circle (0.01cm);
 \filldraw (0.796125, 0.399643) circle (0.01cm);
 \filldraw (0.658114, 0.568048) circle (0.01cm);
 \filldraw (0.601584, 0.66728) circle (0.01cm); 
\filldraw (0.783496, 0.00141272) circle (0.01cm);
 \filldraw (0.433911, 0.317973) circle (0.01cm);
 \filldraw (0.159228, 0.53165) circle (0.01cm); 
\filldraw (0.328435, 0.605131) circle (0.01cm); 
\filldraw (0.349086,  0.217074) circle (0.01cm); 
\filldraw (0.107937, 0.207642) circle (0.01cm);
 \filldraw (0.617916, 0.986284) circle (0.01cm);
 \filldraw (0.362148, 0.838149) circle (0.01cm);
 \filldraw (0.445535, 0.768103) circle (0.01cm);
 \filldraw (0.700766, 0.366559) circle (0.01cm);
 \filldraw (0.806463, 0.204902) circle (0.01cm);
 \filldraw (0.0918888, 0.761954) circle (0.01cm);
 \filldraw (0.0242553, 0.269427) circle (0.01cm);
 \filldraw (0.468364, 0.491061) circle (0.01cm);
 \filldraw (0.783985, 0.773086) circle (0.01cm);
 \filldraw (0.253241, 0.805743) circle (0.01cm);
 \filldraw (0.37167, 0.521583) circle (0.01cm);
 \filldraw (0.16651, 0.71135) circle (0.01cm);
 \filldraw (0.798919, 0.593356) circle (0.01cm);
 \filldraw (0.701959, 0.933628) circle (0.01cm);
 \filldraw (0.342613, 0.737541) circle (0.01cm);
 \filldraw (0.0892862, 0.48232) circle (0.01cm);
 \filldraw (0.556831, 0.568482) circle (0.01cm); 
 \filldraw (0.954616, 0.885814) circle (0.01cm);
\end{tikzpicture}
\caption{A local minimizer of the energy functional  $f(x-y) = \exp(-N^{\frac{2}{d}}\|x-y\|^2)$ on $\mathbb{T}^2$ (picture taken from \cite{stein3}).} 
\end{figure}
\end{center}

\subsection{Consequence for Regularity} The purpose of this section is to discuss implications of such a regularity statement: usually, such statements are given in terms of purely
spatial properties (i.e. point separation, points being spread out, the empirical measure converging weakly to the uniform measure etc). Here, the regularity property is phrased on
the Fourier side; we discuss implications for the spatial properties as well as their properties when used in numerical integration. \\

\textit{1. Integration error.} While there is a vast literature on using minimizing configurations of interacting energy as sample points for quadrature, our argument provides a direct result
 for $L^2-$functions with compact support in frequency space.
\begin{corollary} We have the identity
$$ \sup_{\supp(\widehat f) \subset B(0,X) } \left\| \int_{\mathbb{T}^d}{f(x) dx} - \frac{1}{N} \sum_{n=1}^{N}{f(x_n)}\right\|_{L^2} =  \frac{ \|f\|_{L^2}}{N}\left(\sum_{\|k\| \leq X \atop k \neq 0}{ \left| \sum_{n=1}^{N}{ e^{2 \pi i \left\langle k, x_n \right\rangle}}\right|^2}\right)^{\frac{1}{2}}.$$
As a consequence, minimal configuration of an interacting energy functional like in Theorem 1 will have optimal error rates for the numerical integration of functions
$$\left\{f \in L^2(\mathbb{T}^d): \supp(\widehat f) \subset B(0,cN^{1/d}) \right\}.$$
\end{corollary}
\begin{proof}
The argument follows quickly from an expansion in Fourier series: the frequency $k=0$ cancels with the integral
\begin{align*}
 \frac{1}{N} \sum_{i=1}^{N}{f(x_i)}  - \int_{\mathbb{T}^d}{f(x) dx}    &=   \frac{1}{N} \sum_{n=1}^{N} \sum_{ \|k\| \leq X}{\widehat{f}(k) e^{2\pi i \left\langle k, x_n\right\rangle}} -  \int_{\mathbb{T}^d}{f(x) dx}  \\
                                                                      &=\sum_{\|k\| \leq X \atop k \neq 0}{\widehat{f}(k)   \frac{1}{N} \sum_{n=1}^{N}  e^{2\pi i \left\langle k, x_n\right\rangle}} 
\end{align*}
A simple application of the Cauchy-Schwarz inequality yields
\begin{align*}
  \sum_{\|k\| \leq X \atop k \neq 0}{\widehat{f}(k)   \frac{1}{N} \sum_{n=1}^{N}  e^{2\pi i \left\langle k, x_n\right\rangle}}  &\leq \biggl(\sum_{\|k\| \leq X \atop k \neq 0}{|\widehat{f}(k)|^2}\biggr)^{\frac{1}{2}} \left(\sum_{\|k\| \leq X \atop k \neq 0}{  \left| \frac{1}{N} \sum_{n=1}^{N}  e^{2\pi i \left\langle k, x_n\right\rangle} \right|^2 } \right)^{\frac{1}{2}} \cr
  &= \frac{ \|f\|_{L^2}}{N}\left(\sum_{\|k\| \leq X \atop k \neq 0}{ \left| \sum_{n=1}^{N}{ e^{2 \pi i \left\langle k, x_n \right\rangle}}\right|^2}\right)^{\frac{1}{2}}.
\end{align*}
Equality for the supremum then follows from $L^2-$duality (i.e. picking $\widehat{f}(k)$ so as to obtain equality in the Cauchy-Schwarz inequality). The fact that this is optimal up to constants follows from \cite{stein1}.
\end{proof}

\textit{2. Spatial regularity.} Almost orthogonality to low-frequency trigonometric polynomials has obvious implications for spatial properties. We refer especially to
classical discrepancy theory where usually the other route is highlighted: notions of irregularity of distribution are bounded from above by Fourier-analytic quantities; we refer
especially to Erd\H{o}s-Turan-Koksma inequality and the books \cite{dick, drmota, kuipers, mont} for an overview. Here we shall merely focus on the notion of $L^2-$discrepancy and highlight the connection to our result; for simplicity, we restrict our attention to the two-dimensional case $d=2$. Let $S \subset \mathbb{T}^2$ be a measurable set and define the discrepancy of a point set $\left\{x_1, \dots, x_N\right\} \subset \mathbb{T}^2$ with respect to $S$ as
$$ D(S) = \# \left\{1 \leq i \leq N: x_i \in S\right\} - N |S|.$$
For any fixed $S$ this quantity may be quite small (say $\leq 1$) even if the points are distributed in a rather irregular fashion. It therefore makes sense to consider the quantity over an entire family of sets and one fairly canonical approach is to simply consider all possible translations of $S$ and define
$$d_{S}(x) := D(S + x) = D(\left\{s + x: s \in S\right\}).$$
 A simple computation shows that we can express the averaged square error in terms of Fourier coefficients as
$$ \int_{\mathbb{T}^2}{ d_{S}(x)^2 dx} = \sum_{k \neq 0}{ |\widehat{\chi_{S}}(k)|^2  \left| \sum_{n=1}^{N}{ e^{2\pi i \left\langle k, x_n\right\rangle}} \right|^2},$$
where $\chi_{S}$ is the characteristic function of the set $S$.
\begin{center}
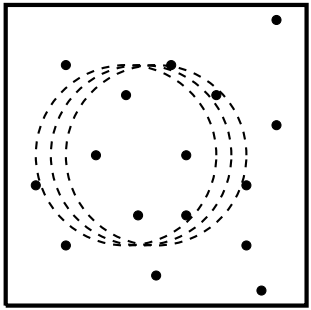
\begin{figure}[h!]
\begin{tikzpicture}[scale=4]
\draw [ultra thick] (0,0) -- (1,0) -- (1,1) -- (0,1) -- (0,0);
\filldraw (0.1, 0.4) circle (0.015cm);
\filldraw (0.2, 0.2) circle (0.015cm);
\filldraw (0.2, 0.8) circle (0.015cm);
\filldraw (0.3, 0.5) circle (0.015cm);
\filldraw (0.4, 0.7) circle (0.015cm);
\filldraw (0.44, 0.3) circle (0.015cm);
\filldraw (0.5, 0.1) circle (0.015cm);
\filldraw (0.55, 0.8) circle (0.015cm);
\filldraw (0.6, 0.5) circle (0.015cm);
\filldraw (0.8, 0.4) circle (0.015cm);
\filldraw (0.7, 0.7) circle (0.015cm);
\filldraw (0.6, 0.3) circle (0.015cm);
\filldraw (0.9, 0.6) circle (0.015cm);
\filldraw (0.9, 0.95) circle (0.015cm);
\filldraw (0.85, 0.05) circle (0.015cm);
\filldraw (0.8, 0.2) circle (0.015cm);
\draw [thick, dashed] (0.4, 0.5) circle (0.3cm);
\draw [thick, dashed] (0.45, 0.5) circle (0.3cm);
\draw [thick, dashed] (0.5, 0.5) circle (0.3cm);
\end{tikzpicture}
\caption{$L^2-$discrepancy: moving a shape over the points and integrating the square of the difference between the number of points inside the shape and the expected number.}
\end{figure}
\end{center}
This quantity depends on both the behavior of the Fourier transform of the characteristic function of $S$ as well as the behavior of the
exponential sum; our result guarantees that the sum over the second term up to frequencies $\| k \| \lesssim N^{1/d}$ will be small, which is the best possible
behavior one could ask of that expression.\\

In this spirit, we introduce another natural regularity measure:
instead of taking the discrepancy with respect to the characteristic
function of a set, we may define discrepancy with respect to a
localized measure; more precisely, given a set of points
$\left\{x_1, \dots, x_N\right\} \subset \mathbb{T}^d$, a parameter
$t> 0$ and a fixed point $x \in \mathbb{T}^d$, we can define the
\textit{heat discrepancy} via
$$ d_{t,\Delta}(x) := \sum_{n=1}^{N}{ \left[e^{t \Delta}\delta_x\right](x_n)} - N.$$
Here and henceforth, $e^{t \Delta}$ denotes the heat propagator defined via
$$ e^{t \Delta}f = \sum_{k \in \mathbb{Z}^d}{e^{-t \|k\|^2} \widehat{f}(k) e^{2 \pi i \left\langle k, x \right\rangle}}.$$
It corresponds to the solution of the heat equation at time $t/(4\pi^2)$ and can be simultaneously understood
as a Fourier multiplier and mollification operator.
Note that $e^{t \Delta}\delta_x$ in particular may be understood as, roughly, a
Gaussian centered at $x$ at carrying most of its mass at scale
$\sim t^{1/2}$. Moreover, the function $e^{t \Delta} \delta_x$ is
scaled so as to have total integral $1$: subtracting $N$ then yields
a function that has integral 0 when integrated over $\mathbb{T}^d$ and
its deviation from $0$ (\textit{i.e.}, its $L^2-$norm) serves as a
natural measure of irregularity.  We would expect this function to be
roughly at order $\sim N$ in most points for $t \sim N^{-\frac{2}{d}}$
independently of the set of points and then to become more regular as
$t$ increases. We will show in the next Corollary that the point sets
for Riesz-type potentials as in Theorem 1 have good regularity
properties with respect to the heat discrepancy.

\begin{corollary} We have, for any set of points,
$$ \int_{\mathbb{T}^d}{ d_{N^{-\frac{2}{d}},\Delta}(x)^2 dx} \gtrsim N^2 .$$
If $\left\{x_1, \dots, x_N\right\} \subset \mathbb{T}$ is an admissible point set in the sense of satisfying the assumptions of Theorem 1 and if, additionally,
$$|\Delta f| \lesssim \frac{|f|}{\|x\|^2} \quad \mbox{and} \quad  \sum_{i \neq j}{f(x_i - x_j)} \leq N^2 \int_{\mathbb{T}^d}{\int_{\mathbb{T}^d}{f(x-y) dx}dy} + c N^{1 + \frac{s}{d}},$$
 then
$$ \int_{\mathbb{T}^d}{ d_{t,\Delta}(x)^2 dx} \lesssim_{\alpha, d}   \frac{(\log{N})^{d-s}}{t^{\frac{d-s}{2}}} \left(  t \sum_{i \neq j}{\frac{1}{\|x_i - x_j\|^{s+2}}} + N^{1+\frac{s}{d}} \right).$$
\end{corollary}
This shows the emergence of some degree of regularity in the regime $s \leq d - 2$
for $t$ large. The inequality raises a natural question that may be new: is it the case
that points that have small $s-$Riesz energy naturally also have a small $t-$Riesz
energy for $t > s$? Since one would expect minimizing configurations to be close to
lattices, this does not seem unreasonable.
The question of obtaining upper bounds on the heat discrepancy is similar in spirit to recent results
\cite{chatterjee, serf} and references therein on \emph{rigidity} of
minimizing point configurations (the points are more regularly
distributed than i.i.d.~random).

\subsection{Open questions.} These results motivate many questions.

\begin{enumerate}
\item We do not know whether
  $f(x) = \|x\|^{-s}$ restricted on torus for $0 < s < d$ satisfies
  the assumptions of the Theorem 1 (the missing property that would need
  to be established being $\widehat{f}(k) \gtrsim (1+\|k\|)^{d-s}$, we
  refer to Hare \& Roginskaya \cite{hare} for results in this
  direction).
\item It would be interesting to understand the behavior of the Fourier-analytic quantity for other cutoff-values. Is it possible to prove bounds on
 $$ \sum_{\|k \| \leq X \atop k \neq 0}{ \left| \sum_{n=1}^{N}{ e^{2 \pi i \left\langle k, x_n \right\rangle}}\right|^2} \qquad \mbox{for arbitrary}~X?$$
We do not know what happens for $X \lesssim N^{\frac{1}{d}}$. For $X \gtrsim N^{\frac{1}{d}}$ it is likely that
 $$ \sum_{\|k \| \leq X \atop k \neq 0}{  \left| \sum_{n=1}^{N}{ e^{2 \pi i \left\langle k, x_n \right\rangle}}\right|^2} \sim N X^d$$
since this would be implied by the conjecture that the points are maximally separated $\|x_i - x_j\| \gtrsim N^{-1/d}$ (see \cite{stein1} for a proof of this implication).
\item Our argument is heavily based on properties of Fourier series while the underlying problem should actually display fairly universal behavior on
arbitrary compact manifolds. It could be interesting to see whether similar results 
hold true on the sphere $\mathbb{S}^{d-1}$, where spherical harmonics provide a fairly accessible function basis to work with. An encouraging Montgomery-type result
on the sphere was established by Bilyk \& Dai \cite{bilyk} (see also \cite{new}).
\end{enumerate}

\section{Proofs}

\begin{lemma}\label{lem:heatkernel} Let $0 < s < d$ be fixed. Then, for all $x,y \in \mathbb{T}^d$ and all $ t > 0$,
$$ \int_{\mathbb{T}^d}{\int_{\mathbb{T}^d}{    \frac{  \left[e^{t\Delta}\delta_x\right](a)  \left[e^{t\Delta}\delta_y\right](b) }{\| a - b \|^s} da db}} \leq \min \left\{ \frac{c_{d,s}}{\|x - y\|^s}, \frac{1}{\|x - y\|^s} + \frac{c_{d,s}' t}{\|x-y\|^{s+2}}\right\}.$$
\end{lemma}
\begin{proof} We prove the result on Euclidean space $\mathbb{R}^d$, the result on the Torus then follows by transplantation for short times and is easily seen to be true
for large times (for large times, the left-hand side converges to a universal constant). On $\mathbb{R}^d$, we can make explicit use of the fact that convolution with the
Riesz potential $\|x\|^{-s}$ is a Fourier multiplier $R$. Fourier multipliers commute and thus
\begin{align*}
  \int_{\mathbb{R}^d}{\int_{\mathbb{R}^d}{    \frac{  \left[e^{t\Delta}\delta_x\right](a)  \left[e^{t\Delta}\delta_y\right](b) }{\| a - b \|^s} da db}} &=   \int_{\mathbb{R}^d}{ \left(\int_{\mathbb{R}^d}{ \frac{  \left[e^{t\Delta}\delta_x\right](a) }{\| a - b \|^s}  da } \right)\left[e^{t\Delta}\delta_y\right](b) db}\\
&= \left\langle R e^{t\Delta} \delta_{x}, e^{t\Delta} \delta_{y} \right\rangle \cr
  &= \left\langle R e^{2t\Delta} \delta_{0}, \delta_{y-x} \right\rangle
\end{align*}
This last expression is completely explicit, $e^{2t \Delta} \delta_0$ is a Gaussian centered in the origin to which the Riesz transform is applied; the result then follows from an explicit computation: we obtain
$$  \left\langle R e^{2t\Delta} \delta_{0}, \delta_{y-x} \right\rangle = \frac{1}{(8\pi t)^{d/2}} \int_{\mathbb{R}^d}{  \frac{ e^{-\frac{\|z\|^2}{8t}}}{ \|z - (y-x)\|^{s}} dz}.$$
We abbreviate the unit vector $w = (y-x)/\|y-x\|$ and argue that
$$ \int_{\mathbb{R}^d}{  \frac{ e^{-\frac{\|z\|^2}{8t}}}{ \|z - (y-x)\|^{s}} dz} = \frac{1}{\|y-x\|^{s}} \int_{\mathbb{R}^d}{  \frac{ e^{-\frac{\|z\|^2}{8t}}}{ \|\frac{z}{\|x-y\|} -w\|^{s}} dz}$$
We can rewrite this integral as
$$  \int_{\mathbb{R}^d}{  \frac{ e^{-\frac{\|z\|^2}{8t}}}{ \|\frac{z}{\|x-y\|} -w\|^{s}} dz} =  \int_{\mathbb{R}^d}{  \frac{ e^{-\|\frac{z}{\|x-y\|}\|^2 \frac{\|x-y\|^2}{8t}}}{ \|\frac{z}{\|x-y\|} -w\|^{s}} dz} =\|x-y\|^{-d}
  \int_{\mathbb{R}^d}{ \frac{e^{-\|s\|^2 \frac{ \|x-y\|^2}{8t}}}{ \|s -w\|^{s}} ds}.$$
Altogether, this implies, with the substitution $t^* = t \|x-y\|^2$ that
$$ \left\langle R e^{2t\Delta} \delta_{0}, \delta_{y-x} \right\rangle = \frac{1}{\|x-y\|^s}   \frac{1}{(8 \pi t^*)^{d/2}}\int_{\mathbb{R}^d}{ \frac{e^{- \frac{\|s\|^2}{8 t^*}} }{ \|s -w\|^{s}} ds}$$

We can use rotational invariance to assume that $w = (1,0,0,\dots,0)$.This turns the remaining expression into a function of $t^*$ which is finite for every $t^* \geq 0$,
converges to 0 as $t^* \rightarrow \infty$ and is thus bounded.
\end{proof}

$$$$

\subsection{Proof of Theorem 1}

\begin{proof}
We assume that the function $f$ satisfies an estimate
$$ \min_{x_1, \dots, x_N} \sum_{i \neq j}{f(x_i - x_j)} - N^2\int_{\mathbb{T}^d}{f(x) dx} \lesssim  N^{1+\frac{s}{d}}.$$
We will henceforth fix $N$ and let $\left\{x_1, \dots, x_N\right\} \subset \mathbb{T}^d$ be a point set for which this inequality is satisfied.
The first step of the argument consists in a mollification of the Dirac point masses: this leads to well-defined functions on which we can apply Fourier analysis. Since $\widehat{f}(k) \geq 0$,
we have that for any $x, y \in \mathbb{T}$ and any measure $\mu$ 
$$\left\langle e^{t \Delta} \mu, e^{t\Delta} (f* \mu) \right\rangle = \sum_{k \in \mathbb{Z}^d}{  e^{-2 t \|k\|^2}  \left| \widehat{\mu}(k) \right|^2 \widehat{f}(k)  }$$
is monotonically decaying in $t$. This suggests using the heat kernel as a mollifier. 
We observe that the self-interactions are at scale 
\begin{align*}
\left\langle e^{t \Delta} \delta_x, e^{t \Delta} (f * \delta_x) \right\rangle &= \left\langle e^{t \Delta} \delta_0, e^{t \Delta} (f * \delta_0) \right\rangle \\
&=  \sum_{k\in \mathbb{Z}^d}{ e^{- 2 t \|k\|^2} \widehat{f}(k)} \lesssim_{c_2} \sum_{k\in \mathbb{Z}^d}{ \frac{e^{- 2 t \|k\|^2}}{1+\|k\|^{d-s}}} \\
&\lesssim  1 + \sum_{1 \leq \|k\| \leq t^{-1/2}}{ \frac{1}{\|k\|^{d-s}}} +  \sum_{|k\| \geq t^{-1/2}}{ \frac{e^{- 2 t \|k\|^2}}{\|k\|^{d-s}}}\\
&\lesssim 1 + \int_{1}^{t^{-1/2}}{\frac{r^{d-1}}{r^{d-s}} dr} + \int_{t^{-1/2}}^{\infty}{\frac{ e^{-t r^2} r^{d-1}}{r^{d-s}} dr} \\
&\lesssim 1 + t^{-\frac{s}{2}} +  \int_{t^{-1/2}}^{\infty}{ e^{-t r^2} r^{s-1} dr}.
\end{align*}
This last integral can be bounded after substituting $x= t r^2$
\begin{align*}
  \int_{t^{-1/2}}^{\infty}{ e^{-t r^2} r^{s-1} dr} &= \int_{1}^{\infty}{ e^{-x} \left( \frac{x}{t} \right)^{\frac{s-1}{2}} \frac{dx}{2 \sqrt{x} \sqrt{t}}} \\
&= t^{-\frac{s}{2}}  \int_{1}^{\infty}{ e^{-x} x^{\frac{s}{2}-1} dx} \lesssim_{s} t^{-\frac{s}{2}}.
\end{align*}
Combining this with Lemma~\ref{lem:heatkernel} implies that, for all $ t > 0$,
\begin{align*}
\left\langle e^{t\Delta} \sum_{i=1}^{N}{\delta_{x_i} * f},  e^{t\Delta} \sum_{i=1}^{N}{\delta_{x_i}} \right\rangle &= 
N \left\langle e^{t\Delta}  \delta_0 * f,  e^{t\Delta} \delta_0 \right\rangle +   \sum_{i, j=1 \atop i \neq j}^{N}{ \left\langle e^{t\Delta}\delta_{x_i} * f,  e^{t\Delta}\delta_{x_j} \right\rangle}\\
&\lesssim N t^{-\frac{s}{2}} +  \sum_{i,j=1 \atop i \neq j}^{N}{f(x_i - x_j)} + c t \sum_{i,j=1 \atop i \neq j}^{N}{\left|  (\Delta f)(x_i - x_j)\right| }
\end{align*}

A simple application of the Fourier transform shows
\begin{align*}
  \sum_{i, j = 1}^{N}{\left\langle e^{t\Delta}\delta_{x_i} * f,  e^{t\Delta} \delta_{x_j} \right\rangle}  &= \sum_{k \in \mathbb{Z}^d}{ e^{-2t \|k\|^2}\left| \sum_{n=1}^{N}{ e^{2 \pi i \left\langle k, x_n \right\rangle}}\right|^2 \widehat{f}(k)} \\
&=  N^2 \int_{\mathbb{T}^d}{f(x) dx} + \sum_{k \in \mathbb{Z}^d \atop k \neq 0}{ e^{-2t \|k\|^2} \left| \sum_{n=1}^{N}{ e^{2 \pi i \left\langle k, x_n \right\rangle}}\right|^2 \widehat{f}(k)}.
\end{align*}
As a consequence, we have the estimate
$$ \sum_{k \in \mathbb{Z}^d \atop k \neq 0}{ e^{-2t \|k\|^2} \left| \sum_{n=1}^{N}{ e^{2 \pi i \left\langle k, x_n \right\rangle}}\right|^2 \widehat{f}(k)} \lesssim  N t^{-\frac{s}{2}} +  E + c t \sum_{i,j=1 \atop i \neq j}^{N}{\left|  (\Delta f)(x_i - x_j)\right| }.$$
We set $t = N^{-\frac{2}{d}}$ and obtain $N t^{-\frac{s}{2}}  \sim   N^{1+\frac{s}{d}}$. Moreover, we have $E \gtrsim N^{1+\frac{s}{d}}$ and thus
\begin{align*}
 E  + N^{-\frac{2}{d}} \sum_{i,j=1 \atop i \neq j}^{N}{\left|  (\Delta f)(x_i - x_j)\right| }&\gtrsim  \sum_{k \in \mathbb{Z}^d \atop k \neq 0}{ e^{-2 N^{-\frac{2}{d}} \|k\|^2} \left| \sum_{n=1}^{N}{ e^{2 \pi i \left\langle k, x_n \right\rangle}}\right|^2 \widehat{f}(k)} \\
&\geq  \sum_{\|k \| \leq N^{1/d} \atop k \neq 0}{ e^{-2 N^{-\frac{2}{d}} \|k\|^2}  \left| \sum_{n=1}^{N}{ e^{2 \pi i \left\langle k, x_n \right\rangle}}\right|^2\widehat{f}(k)} \\
&\gtrsim \frac{1}{N^{\frac{d-s}{d}}}  \sum_{\|k \| \leq N^{1/d} \atop k \neq 0}{  \left| \sum_{n=1}^{N}{ e^{2 \pi i \left\langle k, x_n \right\rangle}}\right|^2}.
\end{align*} 
A reformulation yields
$$ \sum_{\|k \| \leq N^{1/d} \atop k \neq 0}{  \left| \sum_{n=1}^{N}{ e^{2 \pi i \left\langle k, x_n \right\rangle}}\right|^2} \lesssim E N^{1 - \frac{s}{d}} + N^{\frac{d-(s+2)}{d}} \sum_{i,j=1 \atop i \neq j}^{N}{\left|  (\Delta f)(x_i - x_j)\right| }.$$
\end{proof}

\subsection{Proof of Theorem 2}
\begin{proof} We simplify the argument from \cite{stein1} for the convenience of the reader.
The proof proceeds along the same lines as before. It is to see, by considering a lattice, that the energy of any minimal-energy configuration is bounded from above by
$$  \sum_{i, j =1 }^{N}{ \exp\left( - N^{\frac{2}{d}} \|x_i - x_j\|^2 \right)}  \lesssim  N\sum_{j=1}^{\infty}{ j^{d-1} \exp(-j^2)} \lesssim N.$$
 We start by remarking that we can use the short-time asymptotic of the heat kernel to write
$$ \exp\left( - N^{\frac{2}{d}} \|x_i - x_j\|^2 \right) = (1+o(1))\pi^{d/2} \frac{1}{N} \left[e^{(N^{2/d}/4)\Delta}\delta_{x_i}\right](x_j).$$
This allows us to write
\begin{align*}
c N^2 &\gtrsim N\sum_{i,j = 1}^{N}\exp\left( - N^{\frac{2}{d}} \|x_i - x_j\|^2 \right) \gtrsim \sum_{i, j=1}^{N}  \left[e^{(N^{2/d}/4)\Delta}\delta_{x_i}\right](x_j)\\
&\gtrsim   \sum_{k \in \mathbb{Z}^d}{ e^{-\frac{4 \|k\|^2}{N^{2/d}}}  \left| \sum_{n=1}^{N}{ e^{2 \pi i \left\langle k, x_n \right\rangle}}\right|^2} \geq N^2 +  \sum_{k \neq 0}{ e^{-\frac{4 \|k\|^2}{N^{2/d}}}  \left| \sum_{n=1}^{N}{ e^{2 \pi i \left\langle k, x_n \right\rangle}}\right|^2}\\
&\gtrsim   \sum_{\|k \| \leq N^{1/d} \atop k \neq 0}{  \left| \sum_{n=1}^{N}{ e^{2 \pi i \left\langle k, x_n \right\rangle}}\right|^2}
\end{align*}
which is the desired result.
\end{proof}

\subsection{Proof of Corollary 2}
\begin{proof} We observe that
$$ d_{t, \Delta}(x) = \sum_{n=1}^{N}{ \left\langle e^{t \Delta} \delta_x, \delta_{x_n} \right\rangle} - N = \left\langle \delta_x,  e^{t \Delta}  \sum_{n=1}^{N}{  \delta_{x_n}}  \right\rangle - N$$
and thus, by taking the Fourier transform and Plancherel's theorem,
$$  \int_{\mathbb{T}^d}{ d_{t,\Delta}(x)^2 dx} = \sum_{k \in \mathbb{Z}^d \atop k \neq 0}{ e^{- t \|k\|^2} \left| \sum_{n=1}^{N}{ e^{2 \pi i \left\langle k, x_n \right\rangle}}\right|^2}.$$
For every set of points and every $t = N^{-2/d}$, we have
$$  \sum_{k \in \mathbb{Z}^d \atop k \neq 0}{ e^{- t \|k\|^2} \left| \sum_{n=1}^{N}{ e^{2 \pi i \left\langle k, x_n \right\rangle}}\right|^2} \gtrsim  \sum_{ \|k\| \leq N^{\frac{1}{d}} \atop k \neq 0}{  \left| \sum_{n=1}^{N}{ e^{2 \pi i \left\langle k, x_n \right\rangle}}\right|^2} \gtrsim N^2,$$
where the last step is Montgomery's Lemma. 
Conversely, for any set of points under consideration, we have
\begin{align*}
\int_{\mathbb{T}^d}{ d_{t,\Delta}(x)^2 dx} &\lesssim   \sum_{\|k\| \leq X \atop k \neq 0}{e^{-t \|k\|^2} \left| \sum_{n=1}^{N}{ e^{2 \pi i \left\langle k, x_n \right\rangle}}\right|^2} +   \sum_{\|k\| \geq X}{ e^{- t \|k\|^2} \left| \sum_{n=1}^{N}{ e^{2 \pi i \left\langle k, x_n \right\rangle}}\right|^2} \\
&\lesssim  X^{d-s}   \sum_{\|k\| \leq X \atop k \neq 0}{ e^{-t \|k\|^2}\left| \sum_{n=1}^{N}{ e^{2 \pi i \left\langle k, x_n \right\rangle}}\right|^2 \widehat{f}(k)}  +  N^2 \sum_{\|k\| \geq X}{ e^{- t \|k\|^2}}.
\end{align*}
This expression has two sums which we bound separately and by different means.\\

\textbf{First sum.}
The first sum is easy to bound from above by simply summing over all $k \in \mathbb{Z}^d, k \neq 0$ (i.e. forgetting about the restriction $\| k \| \leq X$. That full sum can be rewritten as
\begin{align*}
A = \left\langle e^{(t/2)\Delta} \sum_{i=1}^{N}{\delta_{x_i} * f},  e^{(t/2)\Delta} \sum_{i=1}^{N}{\delta_{x_i}} \right\rangle - N^2\int_{\mathbb{T}^d}{f(x) dx}.
\end{align*}
We shall separate self-interactions from other interactions and write
\begin{align*}
A &=  N \left\langle (e^{(t/2)\Delta} \delta_0)*f, e^{(t/2)\Delta} \delta_0 \right\rangle +   \sum_{i,j=1 \atop i \neq j}^{N}{ \left\langle e^{(t/2)\Delta}\delta_{x_i} * f,  e^{(t/2)\Delta} \delta_{x_j} \right\rangle}  - N^2\int_{\mathbb{T}^d}{f(x) dx}.
\end{align*}
The self-interactions were already computed in the proof of Theorem 1 and yield, for $t \gtrsim N^{-2/d}$,
$$N \left\langle (e^{(t/2)\Delta} \delta_0)*f, e^{(t/2)\Delta} \delta_0 \right\rangle \lesssim N^{1+ \frac{s}{d}}.$$
We shall bound the remaining terms by adding 0 in a suitable way
\begin{align*}
B &=   \sum_{i,j=1 \atop i \neq j}^{N}{ \left\langle e^{(t/2)\Delta}\delta_{x_i} * f,  e^{(t/2)\Delta} \delta_{x_j} \right\rangle} -   \sum_{i,j=1 \atop i \neq j}^{N}{ \left\langle \delta_{x_i} * f,  \delta_{x_j} \right\rangle} \\
&+ \sum_{i,j=1 \atop i \neq j}^{N}{ \left\langle \delta_{x_i} * f,  \delta_{x_j} \right\rangle} - N^2\int_{\mathbb{T}^d}{f(x) dx}.
\end{align*}
Assumption \ref{eq:condition2} implies that the second difference is bounded from above
$$  \sum_{i,j=1 \atop i \neq j}^{N}{ \left\langle \delta_{x_i} * f,  \delta_{x_j} \right\rangle} - N^2\int_{\mathbb{T}^d}{f(x) dx} \leq c_3 N^{1 + \frac{s}{d}}.$$
It remains to deal with the first term which we do on a term-by-term basis; it suffices to bound, for $x,y \in \mathbb{T}$,
$  \left\langle e^{(t/2)\Delta}\delta_{x} * f,  e^{(t/2)\Delta} \delta_{y} \right\rangle - f(x-y)$
or, using translation invariance, for an arbitrary $x \in \mathbb{T}$
$$  \left\langle e^{(t/2)\Delta}\delta_{0} * f,  e^{(t/2)\Delta} \delta_{x} \right\rangle - f(x) = \int_{\mathbb{T}^d}{ (f(x + y) - f(x))  \left[e^{t\Delta}\delta_0\right](y) dy}.$$
If $\|x\| \gtrsim \sqrt{t}$, then this sum is fairly easy to bound via a Taylor expansion and we obtain
$$  \int_{\mathbb{T}^d}{ (f(x + y) - f(x))  \left[e^{t\Delta}\delta_0\right](y) dy} \leq  t\left| \Delta f(y)\right| \lesssim \frac{t}{\|y\|^{s+2}}.$$
If $\|x\| \lesssim \sqrt{t}$, then Lemma 1 implies
$$  \int_{\mathbb{T}^d}{ (f(x + y) - f(x))  \left[e^{t\Delta}\delta_0\right](y) dy} \leq \int_{\mathbb{T}^d}{ f(x + y)   \left[e^{t\Delta}\delta_0\right](y) dy} \lesssim \frac{1}{\|x\|^s} \lesssim \frac{t}{\|x\|^{s+2}}.$$

 \textbf{Second Sum.} We can estimate the second sum using the incomplete gamma function 
\begin{align*}
\sum_{\|k\| \geq X}{ e^{- t \|k\|^2}} \lesssim \int_{X}^{\infty}{ e^{-t r^2} r^{d-1}dr} \lesssim_d  t^{-\frac{d}{2}}\Gamma\left(\frac{d}{2}, X^2 t\right)
\end{align*}
and a simple asymptotic estimate, valid for $X \gtrsim t^{-1/2}$, simplifies this to
$$ t^{-\frac{d}{2}}  +  t^{-\frac{d}{2}}\Gamma\left(\frac{d}{2}, X^2 t\right) \lesssim   X^{d-2} t^{-1} e^{- X^2 t}.$$

\textbf{Conclusion.}
Altogether, we have established that
$$ \int_{\mathbb{T}^d}{ d_{t,\Delta}(x)^2 dx} \lesssim_d   X^{d-s} \left(  t \sum_{i \neq j}{\frac{1}{\|x_i - x_j\|^{s+2}}} + N^{1+\frac{s}{d}} \right) +  X^{d-2} t^{-1} N^2 e^{- X^2 t^{}}.$$
We now substitute $X =c t^{-1/2} \log{N}$  for some constant $c \gg 1$. This ensures, combined with $t \gtrsim N^{1/2}$ that the second terms is smaller than the first term and thus
$$ \int_{\mathbb{T}^d}{ d_{t,\Delta}(x)^2 dx} \lesssim_d  \frac{(\log{N})^{d-s}}{t^{\frac{d-s}{2}}} \left(  t \sum_{i \neq j}{\frac{1}{\|x_i - x_j\|^{s+2}}} + N^{1+\frac{s}{d}} \right)$$
\end{proof}

\section*{Compliance with Ethical Standards}
The authors are not aware of any potential conflicts of interest. This
work did not involve any human data or animal participation. We comply
with other ethical standards of the journal.

\end{document}